\documentclass[11pt,letter]{article}
\usepackage[latin1]{inputenc} 
\usepackage{amssymb}
\usepackage[mathscr]{eucal}
\usepackage{psfrag}
\usepackage{latexsym} 
\usepackage{xspace} 
\usepackage{pstricks}
\usepackage{pst-node} 
\newpsobject{showgrid}{psgrid}{subgriddiv=1,griddots=10,gridlabels=6pt}

\usepackage{graphicx} 
\usepackage{amsmath}
\usepackage{enumerate}
\usepackage{subfigure}
\usepackage{amsfonts}
\usepackage{float}
\usepackage{array}
\usepackage{epsfig}
\usepackage{pstricks}
\usepackage{pst-node} 
\usepackage{openbib}
\usepackage{times}
\usepackage{mathptm}
\usepackage[american]{babel}
\usepackage{amsthm}

\usepackage{anysize}

\usepackage{algorithm}
\usepackage{algorithmic}

\newtheorem{theorem}{Theorem}[section]
\newtheorem{lemma}[theorem]{Lemma}

\newtheorem{definition}[theorem]{Definition}

\newtheorem{assumption}[theorem]{Assumption}

\newtheorem{claim}[theorem]{Claim}

\newcommand{\bqed}{$\blacksquare$}

\newcommand{\setI}{{\mathcal I}}

\newcommand{\setC}{{\mathcal C}}

\title{
{\bf On the Recognition of Fuzzy Circular Interval Graphs}}

\author{Gianpaolo Oriolo\footnote{Universit\`a di Roma ``Tor Vergata'', Dipartimento di Ingegneria dell'Impresa, Viale del Politecnico 1, 00133 Roma, Italy. E-mail: \texttt{oriolo@disp.uniroma2.it}} \and
Ugo Pietropaoli\footnote{Universit\`a di Roma ``Tor Vergata'', Dipartimento di Ingegneria dell'Impresa, Viale del Politecnico 1, 00133 Roma, Italy. E-mail: \texttt{pietropaoli@disp.uniroma2.it}} \and 
Gautier Stauffer\footnote{Institut de Math\`ematiques de Bordeaux, Universite de Bordeaux 1. Email:\texttt{gstauffer@math.u-bordeaux1.fr}}}

\date{}

\begin{document}

\maketitle

\begin{abstract}

\noindent 

{\em Fuzzy circular interval graphs} are a generalization of proper circular arc graphs and have been recently introduced by Chudnovsky and Seymour as a fundamental subclass of claw-free graphs. In this paper, 
we provide a polynomial-time algorithm for recognizing such graphs, and more importantly for building a suitable representation. 
\end{abstract}


\section{Introduction}
\label{intro}

A graph is \emph{claw-free} if no vertex has three pairwise non-adjacent neighbors. Claw-free graphs have been receiving much of attention in the last years, especially after Chudnovsky and Seymour (see e.g. \cite{CS_Survey} and \cite{CS_3}) prove several structural results for those graphs.  They show in particular that a claw-free graph is either a \emph{fuzzy circular interval graph}~\cite{CS_Survey} ({\sc fcig} for short, see Def. \ref{def:fuzzy}) or the ``composition" of some base-graphs; moreover, as they point out \cite{CS_3}, ``{\em [fuzzy circular interval graphs] are claw-free, and these together with line graphs turn out to be the two ``principal'' basic classes of claw-free graphs}.''  In fact, {\sc fcig}s also play a crucial role in a linear description of the stable set polytope of {\em quasi-line graphs}, a relevant sub-class of claw-free graphs~\cite{EOSV}. ({\sc fcig}s are also called graphs that are {\em thickening of circular interval trigraphs}, see~\cite{CS_3}.) 

\smallskip
In this paper we shed some light onto the class of {\sc fcig}s. We describe an algorithm for recognizing {\sc fcig}s, and building a suitable representation. While a recognition algorithm could be possibly derived from a characterization of {\sc fcig}s in terms of excluded subgraphs \cite{C_private}, no algorithm for constructing a representation was available before. 

Our idea builds upon the fact that a {\sc fcig} without \emph{proper and homogeneous} pairs of cliques is indeed a \emph{circular interval graph} (definitions come later) and that circular interval graphs (which are also called \emph{proper circular arc graphs}) admit poly-time algorithms for solving the recognition problem \cite{C, DHH, MC}. We therefore introduce an operation of \emph{reduction} of proper and homogeneous pairs of cliques, which preserves the fuzzy circular interval structure. In particular, by applying this operation a polynomial number of times to a graph $G$, we end up with a graph $G'$ without proper and homogeneous pairs of cliques. Moreover $G'$ is circular interval if and only if $G$ was a {\sc fcig}. 
All together, we derive a polynomial-time algorithm to recognize whether a graph is a {\sc fcig}, and, in case, build a suitable representation. In fact, building upon a few facts from the literature, this algorithm can be implemented as to run in $O(n^2m)$-time.

\smallskip
The paper is organized as follows.  In Section \ref{fcig}, we introduce the classes of circular interval graphs, fuzzy circular interval graphs and recall the definitions of proper and homogeneous pairs of cliques. Then in Section \ref{sec:prop_hom} we define almost proper pairs of cliques and show some properties of pairs of cliques that are almost proper and homogeneous in a {\sc fcig}. In Section \ref{sec:characterization} we define a reduction operation for homogeneous pairs of cliques, and we prove that this reduction preserves the property of a graph to be a {\sc fcig} when the pair of cliques is proper. Finally, in Section \ref{sec:recognition}, we show how to find and reduce pairs of cliques that are proper and homogeneous, and we provide the recognition and representation algorithm for fuzzy circular interval graphs.

\smallskip
We close the introduction with a definition. A graph $G=(V,E)$ will always be simple and undirected. We denote by $n$ the number of vertices and by $m$ the number of edges. For a set $X\subseteq V$, we denote by $G[X]$ the subgraph induced by $X$. For a vertex $v$, we denote by $N(v)$ the {\em neighborhood} of $v$, i.e.~the set of vertices that are adjacent to $v$.

\begin{definition}
\label{complete}
Let $Q$ be a clique of $G=(V,E)$ and let $v \in V \setminus Q$:
\begin{itemize}
\item $v$ is {\em complete} to $Q$ if $Q\subseteq N(v)$, and $\Gamma(Q)$ is the set of vertices that are complete to $Q$.
\item $v$ is {\em anti-complete} to $Q$ if $Q\cap N(v) = \emptyset$,  and $\overline\Gamma(Q)$ is the set of vertices that are anti-complete to $Q$.
\item $v$ is {\em proper} to $Q$ if $v$ is neither complete nor anti-complete to $Q$, and $P(Q)$ is the set of vertices that are proper to $Q$.
\end{itemize}
\end{definition}

\section{Fuzzy circular interval graphs}
\label{fcig}

Fuzzy circular interval graphs, that are also called graphs that are {\em thickening of circular interval trigraphs} (see e.g.~\cite{CS_3}), have been introduced by Chudnovsky and Seymour as a generalization of the simpler class of circular interval graphs.

\begin{definition} {\em \cite{CS_Survey}}
A {\em circular interval graph} $G=(V,E)$ is  defined by
the following construction: Take a circle ${\mathcal C}$ and a set of vertices
$V$ on the circle. Take a subset of intervals ${\mathcal I}$ of ${\mathcal C}$
and say that $u,v\in V$ are adjacent if $\{u,v\}$ is a subset of one of
the intervals.
\end{definition}

Circular interval graphs (see Figure \ref{fig:cig_fcig}) are also called {\em proper circular arc graphs}, i.e.~they are equivalent to the intersection graphs of arcs of a circle with no containment between arcs~\cite{CS_Survey}. Therefore, we may associate with a circular interval graph both an {\em interval representation} and an {\em arc representation}.

Given a graph $G$ with $n$ vertices and $m$ edges, there are many polynomial time algorithms that recognize whether $G$ is a proper circular arc graph (and therefore a circular interval one) and, in case, build the arc representation (see e.g. \cite{C, DHH, MC}). In this paper, we mainly refer to the linear (i.e.  $\mathcal{O}(n+m)$) time algorithm in \cite{DHH}, since it can be trivially adapted to build, still in linear time, the {\em interval} representation for $G$, if any (see Proposition 2.6 in \cite{DHH}); note also that this representation uses $n$ intervals. 

\begin{definition}\label{def:fuzzy} {\em \cite{CS_Survey}}
A graph $G=(V,E)$ is  {\em fuzzy
circular interval} ({\sc fcig}) if the following conditions hold. 
\begin{itemize}
\item[(i)] There is a map $\Phi$ from $V$ to a circle ${\mathcal C}$.
\item[(ii)] There is a set ${\mathcal I}$ of intervals of ${\mathcal C}$, each homeomorphic to the closed interval [0, 1] and none
including another,  such that no point of ${\mathcal C}$ is the end of more
than one interval and:
\begin{itemize}

\item[(a)] If two vertices $u$ and $v$ are adjacent, then $\Phi(u)$ and
  $\Phi(v)$ belong to a common interval.
\item[(b)]  If two vertices $u$ and $v$ belong to the same interval, which is not an interval with endpoints $\Phi(u)$ and $\Phi(v)$, then they are adjacent.
\end{itemize}
\end{itemize}
\end{definition}

In other words, in a {\sc fcig}, adjacencies are
completely described by the pair $(\Phi, {\mathcal I})$, except  for vertices $u$
and $v$ such that one of the intervals with endpoints $\Phi(u)$ and $\Phi(v)$ belongs to ${\mathcal I}$. For these vertices adjacency is fuzzy (see Figure \ref{fig:cig_fcig}) i.e.~the adjacencies can be arbitrarily chosen. In the following, when referring to a {\sc fcig}, we often consider a \emph{representation} $(\Phi, {\mathcal I})$ and detail the fuzzy adjacencies only when needed. Sometimes, we abuse notation and let $G~=~(V,\Phi,{\mathcal I})$ be a {\sc fcig} with vertex set $V$ and representation $(\Phi, {\mathcal I})$ and again detail the fuzzy adjacencies only when needed. Note that, by definition, each interval of a representation $(\Phi, {\mathcal I})$ of a {\sc fcig} has non-empty interior. It is also easy to see that, if we are given for some {\sc fcig} $G$ a representation $(\Phi,{\mathcal I})$ such that $|{\mathcal I}|> n$, then there is some interval $I\in \setI$ such that $(\Phi,{\mathcal I}\setminus I)$ is still a representation for $G$. Also, as we discussed above, with a trivial modification, the algorithm in \cite{DHH} returns a representation for a {\sc cig} with $n$ intervals. Since our main result, an algorithm for recognizing and building a representation for {\sc fcig}s, builds upon this latter algorithm, in this paper we assume the following:

\begin{assumption}\label{number_of_int}
When we deal with a {\sc fcig} $G$ for which a representation $(\Phi,{\mathcal I})$ is given, we always assume that $|{\mathcal I}|\leq n$.
\end{assumption}

\smallskip
Given a circle ${\mathcal C}$, let $a$ and $b$ be two points of ${\mathcal C}$. We denote by $[a,b]$ the interval of ${\mathcal C}$ that we span if we move clockwise from $a$ to $b$. Similarly $(a,b)$ denotes $[a,b]\setminus \{a,b\}$. Given a point $p$ of $\mathcal{C}$, we denote by $\Phi^{-1}(p)$ the set $\{ v \in V \mid \Phi(v)=p\}$ (note that $\Phi^{-1}(p)$ is a clique if the graph is connected), by $\Phi^{-1}([a,b])$ the set $\{v \in V:  \Phi(v) \in [a,b] \}$, by $\Phi^{-1}((a,b))$ the set $\{v \in V:   \Phi(v) \in (a,b) \}$, and so on. Sometimes, we abuse notation and we say $a\leq \Phi(v) \leq b$ for $\Phi(v) \in [a,b]$ and similarly $a < \Phi(v) < b$ for $\Phi(v) \in (a,b)$. 
If $[p, q]$ is an interval of ${\mathcal I}$ such that $\Phi^{-1}(p)$ and $\Phi^{-1}(q)$ are both
non-empty, then we call $[p, q]$ a {\it fuzzy interval} and the cliques ($\Phi^{-1}(p), \Phi^{-1}(q))$ a {\it fuzzy pair}. 

\smallskip
Substituting {\it line} for {\it circle} in the two previous definitions allows to define {\em linear interval graphs} and {\em fuzzy linear interval graphs}. 
Linear interval graphs are also called {\em proper} (or {\em unit}) {\em interval graphs} and several algorithms are available for solving the recognition problem \cite{Cetal, HMM, PD}.

\begin{figure}[h]
\begin{center}
\vspace{1.2cm}
\begin{tabular}{ccc}
\psset{unit=.7cm}
\begin{pspicture}(4,4)
\pscircle(2,2){2}

\cnode*(4,2){3pt}{1}
\cnode*(2,0){3pt}{2}
\cnode*(.45,.8){3pt}{3}
\cnode*(3,3.65){3pt}{4}
\cnode*(.9,3.6){3pt}{5}
\cnode*(3.2,4){3pt}{6}
\cnode*(.65,3.95){3pt}{7}
\cnode*(.4,4.25){3pt}{8}
\ncline{-}{4}{5}
\ncline{-}{4}{6}
\ncline{-}{4}{7}
\ncline{-}{4}{8}
\ncline{-}{5}{6}
\ncline{-}{5}{7}
\ncline{-}{6}{7}
\ncline{-}{6}{8}
\ncline{-}{7}{8}
\ncarc[arcangleA=90, arcangleB=90]{-}{5}{8}

\ncline{-}{1}{4}
\ncline{-}{1}{6}

\ncline{-}{1}{2}
\ncline{-}{2}{3}
\ncline{-}{1}{3}
\psarc[linecolor=blue](2,2){2.5}{210}{10}
\psarc[linecolor=blue](2,2){3.5}{59}{125}
\psarc[linecolor=blue](2,2){3}{-10}{70}
\end{pspicture} &&

\hspace{2cm}
\psset{unit=.7cm}
\begin{pspicture}(4,4)

\pscircle(2,2){2}

\cnode*(4,2){3pt}{1}
\cnode*(2,0){3pt}{2}
\cnode*(.45,.8){3pt}{3}
\cnode*(3,3.65){3pt}{4}
\cnode*(.9,3.6){3pt}{5}
\cnode*(3.2,4){3pt}{6}
\cnode*(.65,3.95){3pt}{7}
\cnode*(.4,4.25){3pt}{8}

\ncline[linecolor=red, linestyle=dashed]{-}{4}{5}
\ncline{-}{4}{6}
\ncline[linecolor=red, linestyle=dashed]{-}{4}{7}
\ncline[linecolor=red, linestyle=dashed]{-}{4}{8}
\ncline[linecolor=red, linestyle=dashed]{-}{5}{6}
\ncline{-}{5}{7}
\ncline[linecolor=red, linestyle=dashed]{-}{6}{7}
\ncline[linecolor=red, linestyle=dashed]{-}{6}{8}
\ncline{-}{7}{8}
\ncarc[arcangleA=90, arcangleB=90]{-}{5}{8}

\ncline{-}{1}{4}
\ncline{-}{1}{6}

\ncline{-}{1}{2}
\ncline{-}{2}{3}
\ncline{-}{1}{3}

\psarc[linecolor=blue](2,2){2.5}{210}{10}
\psarc[linecolor=blue](2,2){3}{-10}{70}

\psarc[linecolor=red](2,2){3.5}{59}{125}

\end{pspicture}\\
\end{tabular}
\end{center}
\caption{A circular interval graph (on the left) and a fuzzy circular interval graph (on the right). Dashed lines represent fuzzy adjacencies.}
\label{fig:cig_fcig}
\end{figure}
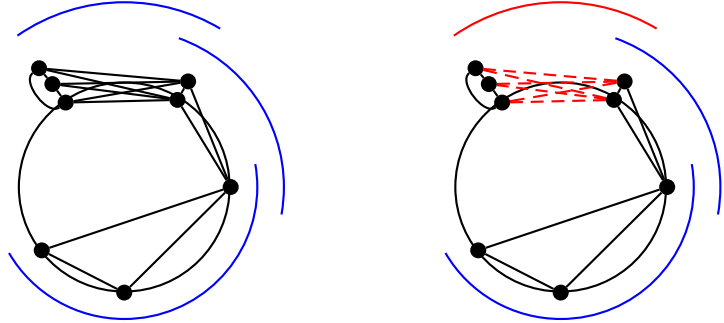

We now recall the definitions of proper and homogeneous pairs of cliques.

\begin{definition}
\label{homo}
Given a graph $G$, a {\em homogeneous} pair of cliques of $G$ is a pair of non-empty vertex disjoint cliques $(K_1, K_2)$ with the property that each $z \not \in (K_1 \cup K_2)$ is either 
complete or anti-complete to $K_1$ and either complete or anti-complete to $K_2$, that is, $z\in (\Gamma(K_1)\cup \overline\Gamma(K_1))\cap (\Gamma(K_2)\cup \overline\Gamma(K_2))$.
\end{definition}

\begin{definition}
\label{proper}
Given a graph $G$, a {\em proper} pair of cliques of $G$ is a pair of non-empty vertex disjoint cliques $(K_1, K_2)$  with the property that each vertex $u \in K_1$ $(K_2$, respectively$)$ is proper [ see Def. \ref{complete} ] to $K_2$ $(K_1)$.
\end{definition}

Proper and homogeneous pairs of cliques (see Fig.~\ref{fig:reduction2} (left)) are called non-trivial homogeneous pairs of cliques in \cite{kr}. The following lemmas, whose simple proofs we skip, link proper and homogeneous pairs of cliques to fuzzy circular interval graphs (the proofs are constructive and rely on the following fact: given a fuzzy pair $(K_1,K_2)$, if $v\in K_1$ is not proper to $K_2$, $v$ is either anti-complete or complete to $K_2$ and one can slightly move $\Phi(v)$ inside or outside the fuzzy interval and recover adjacencies by adding suitable intervals)

\begin{lemma}{\em \cite{EOSV}}
\label{fuzzy1}
Let $G=(V, \Phi, \mathcal{I})$ be a fuzzy circular interval graph. One may build in $O(n^2)$-time a representation for $G$ where each fuzzy pair of cliques is proper and homogeneous.
\end{lemma}

\begin{lemma}
\label{fuzzy2}
Let $G=(V, \Phi, \mathcal{I})$ be a fuzzy circular (resp.~linear) interval graph. If $G$ has no proper and homogeneous pairs of cliques, then $G$ is a circular (resp.~linear) interval graph.
\end{lemma}

\begin{lemma}
\label{building_block} {\em \cite{CS_Survey}}
Let $(K_1, K_2)$ be a proper pair of non-empty cliques of a graph $G$. Then $G[K_1 \cup K_2]$ contains $C_4$ (a chordless cycle of length 4) as an induced subgraph.
\end{lemma}

\section{Almost proper and homogeneous pairs of cliques in {\sc fcig}}
\label{sec:prop_hom}

We start with a slight generalization of the concept of proper pair of cliques.

\begin{definition}
Given a graph $G$, an \emph{almost-proper} pair of cliques of $G$ is a pair of non-empty vertex disjoint cliques $(K_1,K_2)$ with the property that every vertex in $K_1$ (resp. $K_2$) is not complete to $K_2$ (resp. $K_1$) and there exists $u\in K_1$, $v\in K_2$ such that $uv\in E$. 
\end{definition}

\begin{definition}
\label{tight}
Let $G=(V, \Phi, \mathcal{I})$ be a fuzzy circular interval graph and $(K_1, K_2)$ be an almost-proper and homogeneous pair of cliques. We say that $(\Phi, \mathcal{I})$ is {\em tight} with respect to $(K_1, K_2)$ if, for some $w\in K_1$ and $z\in K_2$,  either $[\Phi(w), \Phi(z)]$ or $[\Phi(z), \Phi(w)]$ belongs to $\mathcal{I}$.
\end{definition}

\smallskip
In this section we show the following fundamental fact (Theorem \ref{main}): if we are given a representation $(\Phi, \mathcal{I})$ for some {\sc fcig} $G$, together with an almost-proper and homogeneous pair $(K_1, K_2)$ of cliques of $G$, then in time $O(n^2)$ we may build another representation $(\Phi', \mathcal{I}')$, where all the vertices of $K_1$ (resp. $K_2$) ``sit'' on a same point of ${\mathcal C}$. The proof of Theorem \ref{main} builds upon a technical lemma (Lemma \ref{fuzzy}) showing how to build in $O(n^2)$-time for a {\sc fcig} $G=(V, \Phi, \mathcal{I})$ a representation $(\Phi', \mathcal{I}')$ that is tight with respect to some given pair of almost-proper and homogeneous cliques.

In its turn, the proof of Lemma \ref{fuzzy} is easier, if we first discuss one special case that arises when $G$ has small stability number. This motivates the following definition:

\begin{definition}
\label{fuzzy_dom}
Let $G=(V, E)$ be a graph and $(K_1,K_2)$ a homogeneous pair of cliques. Let $S_3$ be the set of vertices that are complete to both $K_1$ and $K_2$, $S_1$ (resp. $S_2$) the set of vertices  complete to $K_1$ (resp. $K_2$) and anti-complete to $K_2$ (resp. $K_1)$. We say that $(K_1,K_2)$ is a {\em fuzzy dominating pair} if $V = K_1\cup K_2 \cup S_1 \cup S_2 \cup S_3$; $S_1$ and $S_2$ are cliques and are complete to $S_3$ (all remaining adjacencies being possible).
\end{definition}

\begin{lemma}
\label{fuzzy_simple}
Let $G=(V, \Phi, \mathcal{I})$ be a connected fuzzy circular interval graph and $(K_1, K_2)$ a homogeneous pair of cliques. In time $O(n^2)$ we can recognize whether $(K_1, K_2)$ is a fuzzy dominating pair and, in this case, build a representation $(\Phi'', \mathcal{I}'')$, such that: $\Phi''(v) = a$, for each $v\in K_1$; $\Phi''(v) = b$, for each $v\in K_2$;  $[a,b]$ or $[b,a] \in \mathcal{I}''$, for any $a\neq b$ on $\mathcal C$. 
\end{lemma}
\begin{proof} 
Note that, since $(K_1,K_2)$ is a homogeneous pair, $S_1,S_2$ and $S_3$ can be built in time $O(n)$. In order to check that $(K_1, K_2)$ is a fuzzy dominating pair we then need to check that $V = K_1\cup K_2 \cup S_1 \cup S_2 \cup S_3$, that $S_1$ and $S_3$ are cliques and that $S_1, S_2$ are complete to $S_3$. Trivially, that can be done in time $O(n^2)$.

Now suppose that $(K_1, K_2)$ is a fuzzy dominating pair. Every vertex in $K_1\cup K_2$ is complete to $S_3$; therefore $S_3$ can be partitioned into two cliques $S_4$ and $S_5$ (every fuzzy circular interval graph is \emph{quasi-line} i.e. the neighborhood of any vertex can be partitioned into two cliques), that can be found in time $O(n^2)$. Now the sets $K'_1=S_1 \cup K_1 \cup S_4$ and $K'_2=S_2 \cup K_2 \cup S_5$ are cliques and we can therefore represent $G$ with $(\Phi'', \mathcal{I}'')$ where $\Phi''(v) = a$, for each $v\in K'_1$, $\Phi''(v) = b$, for each $v\in K'_2$ and  $\mathcal{I}''= \{[a,b]\}$ for any $a\neq b$ on the circle.
\end{proof} 

Before going to the proof of Lemma \ref{fuzzy}, we need a few more definitions and a lemma.

\begin{definition}
\label{covercircle}
Three intervals $I_1$, $I_2$ and $I_3$ of ${\mathcal C}$ {\em cover} ${\mathcal C}$ if there exist points $a, b, c$ on ${\mathcal C}$ such that $[a,b]\subseteq I_1$, $[b,c]\subseteq I_2$, $[c,a]\subseteq I_3$.
\end{definition}

\begin{definition}
\label{covercirclebis}
Let $G=(V, \Phi, \mathcal{I})$ be a fuzzy circular interval graph and $Q\subseteq V$. We say that an interval $I \in \mathcal{C}$ {\em covers} $Q$ if $\bigcup_{v\in Q} \Phi(v)\subseteq I$.
\end{definition}

\begin{lemma}
\label{interval}
Let $G=(V, \Phi, \mathcal{I})$ be a fuzzy circular interval graph and $K$ a clique of size two or more. Then, either there exists an interval $I \in \mathcal{I}$ covering $K$, or there exist three intervals $I_1$, $I_2$ and $I_3\in \mathcal{I}$ covering the circle. In the latter case, no vertex of $V\setminus K$ is anti-complete to $K$.
\end{lemma}

\begin{proof}
The proof is by induction on the size of $K$. If $|K|=2$, there exists an interval $I \in \mathcal{I}$ covering $K$ by definition of {\sc fcig}s. Now let $K$ be such that $|K|>2$ and $v\in K$. By induction, either there exists an interval $I_1 \in \mathcal{I}$ covering $K \setminus v$, or there exist three intervals covering the circle and no vertex of $(V\setminus K) \cup\{v\}$ is anti-complete to $K \setminus v$. In the latter case, the induction is trivial. Analogously, in the former case, the induction is trivial if $\Phi(v)\in I_1$. So suppose to the contrary that $\Phi(v)\notin I_1$, and assume that $I_1 = [a,b]$. Since $v$ is adjacent to all vertices in $K \setminus v$, it easily follows that either there exists $I_2$ containing $(I_1\cap \Phi(K_1)) \cup \{\Phi(v)\}$, and the result follows, or there must exist $I_2$ and $I_3$, such that $[\Phi(v), a]\subsetneq I_2$ and $[b, \Phi(v)]\subsetneq I_3$ (note that e.g. $[\Phi(v), a]\neq I_2$ because no point of ${\mathcal C}$ is the end of more than one interval of ${\mathcal I}$). Note that $I_1,I_2$ and $I_3$ cover ${\mathcal C}$. Finally, we are left with showing that, in this case, no vertex of $V\setminus K$ is anti-complete to $K$. The statement is trivial for any $u\in V\setminus K$ such that $\Phi(u)\in (a, b)$. So assume that $\Phi(u)\notin (a, b)$, and without loss of generality assume that $\Phi(u)\in [\Phi(v), a]$. Since $[\Phi(v), a]\subsetneq I_2$, it follows that $u$ and $v$ are adjacent, which is enough.
\end{proof}

\begin{lemma}
\label{fuzzy}
Let $G=(V, \Phi, \mathcal{I})$ be a connected fuzzy circular interval graph and $(K_1, K_2)$ be an almost-proper and homogeneous pair of cliques. In time $O(n^2)$ we can either recognize that $(\Phi, \mathcal{I})$ is tight with respect to $(K_1, K_2)$, or build for $G$ another representation $(\Phi', \mathcal{I}')$ that is tight with respect to $(K_1, K_2)$. 
\end{lemma}

\begin{proof} 
We assume that $(K_1, K_2)$ is not a fuzzy dominating pair, else we are done by Lemma \ref{fuzzy_simple}. We can recognize whether $(\Phi, \mathcal{I})$ is tight with respect to $(K_1, K_2)$ in time $O(n^2)$ (recall that we are assuming that $|\mathcal{I}|\leq n$). In the following, we therefore assume that $(\Phi, \mathcal{I})$ is {\em not} tight with respect to $(K_1, K_2)$. We also assume that, for every fuzzy interval in ${\mathcal I}$, every vertex mapped at one of the extremities has an adjacent and a non-adjacent vertex mapped at the other extremity (see Lemma \ref{fuzzy1}, the transformation obviously preserves $(\Phi, \mathcal{I})$ {\em not} tight with respect to $(K_1, K_2)$).

\smallskip
We first show that there exist intervals $I_1, I_2 \in \mathcal{I}$ such that $I_1$ covers $K_1$ or $I_2$ covers $K_2$. In fact, from Lemma \ref{interval}, if no interval of ${\mathcal I}$  covers $K_1$, then no vertex of $V\setminus K_1$ is anti-complete to $K_1$; thus, by homogeneity, each vertex $z\in V\setminus (K_1\cup K_2)$ is complete to $K_1$. 
Similarly, if there is no interval covering $K_2$,  then each vertex $z\in V\setminus (K_1\cup K_2)$ is complete to $K_2$. But then $(K_1,K_2)$ is a fuzzy dominating pair with $S_1=S_2=\emptyset$ and $S_3=V\setminus (K_1\cup K_2)$, a contradiction. 

We can thus assume without loss of generality that there exists an interval $I_1 \in \mathcal{I}$ covering $K_1$. We also define $I'_1:=[a_1,b_1]\subseteq I_1$ to be the smallest interval of ${\mathcal C}$ covering $K_1$. (Notice that $I'_1$ might not be an interval of  $\mathcal{I}$.)  Observe that $a_1\neq b_1$. Indeed, otherwise, since there exists $u\in K_1$, $v\in K_2$ such that $uv\in E$, it would follow that there is an interval $I\in {\mathcal I}$ covering $K_1$ and $v$. But because each vertex of $K_2$ is not complete to $K_1$, necessarily either $I=[a_1, \Phi(v)]$ or $I=[\Phi(v), a_1]$, and this contradicts the assumption that $(\Phi, \mathcal{I})$ is not tight with respect to $(K_1, K_2)$. Note that, since $I'_1\subseteq I_1$, a similar argument shows that no vertex $v\in K_2$ is such that $\Phi(v)\in I'_1$. Therefore, we may define $I'_2:=[a_2,b_2]$ to be the smallest interval in $\setC\setminus I'_1$ covering $K_2$; it follows that $I'_1\cap I'_2 = \emptyset$. Also, by similar arguments as above, $a_2\neq b_2$. Now there exists $I_2$ covering $[a_2,b_2]$ because otherwise there would exist an interval containing $[b_2,a_2]$ ($K_2$ is a clique and the vertices that map to $a_2$ and $b_2$ are adjacent) and thus $[a_1,b_1]$ would be in the interior of this interval and some vertices of $K_2$ (e.g. those that map to $a_2$ or $b_2$) would be complete to $K_1$, a contradiction.
 
It is convenient to summarize our results so far in the following:

\begin{claim}\label{recap}
There exist intervals $I_1, I_2 \in \mathcal{I}$ such that $I_1$ covers $K_1$ and $I_2$ covers $K_2$ and with the property that if we let $I'_1:=[a_1,b_1]\subseteq I_1$ be the smallest interval of ${\mathcal C}$ covering $K_1$ and $I'_2:=[a_2,b_2]\subseteq I_2$ be the smallest interval of ${\mathcal C}$ covering $K_2$, then $I'_1\cap I'_2 = \emptyset$. Note that, by definition, for $i = 1,2$, $K_i\cap \Phi^{-1}(a_i)\neq \emptyset$ and $K_i\cap \Phi^{-1}(b_i)\neq \emptyset$.
\end{claim}

\begin{claim}\label{claim:2}
For all $\overline I_1\in \setI$ such that $I'_1\subseteq \overline I_1$, we have  $\overline I_1\cap I'_2=\emptyset$, and, similarly, for all $\overline I_2\in \setI$ such that $I'_2\subseteq \overline I_2$,  we have $\overline I_2\cap I'_1=\emptyset$.
\end{claim}
Let us show that, for each $\overline I_2\in \setI$ with $I'_2\subseteq \overline I_2$, then  $\overline I_2\cap I'_1=\emptyset$. Indeed, otherwise, there exists $v\in K_1:$ $\Phi(v)\in \overline I_2$. We can assume without loss of generality that $v\in \Phi^{-1}(b_1)$ and $[b_1,b_2] \subseteq {\overline I_2}$. But then either we have $v$ or $[a_2,b_2]$ in the interior of ${\overline I_2}$ and in both cases $v$ is complete to $K_2$, a contradiction, or we have $\overline I_2=[b_1,b_2]$, and this contradicts the assumption that $(\Phi, \mathcal{I})$ is not tight with respect to $(K_1, K_2)$.
\bqed

\begin{claim}\label{claim:5} 
If there exist intervals $\overline I_1, \overline I_2, \overline I_3, \overline I_4$ of $\setI$: $\overline I_1\supseteq [a_1,b_1]$; $\overline I_2\supseteq [a_2,b_2]$; $\overline I_3\supsetneq [b_1,a_2]$; $\overline I_4\supsetneq [b_2,a_1]$, then $(K_1,K_2)$ is a fuzzy dominating pair.
\end{claim}
We now show that then $(K_1,K_2)$ would be a fuzzy dominating pair, with $S_1=\Phi^{-1}((a_1,b_1))\setminus K_1$, $S_2=\Phi^{-1}((a_2,b_2)))\setminus K_2$, $S_3=\Phi^{-1}([b_1,a_2]\cup [b_2,a_1])\setminus (K_1\cup K_2)$, a contradiction. Indeed, because of $\overline I_4$, each vertex $v\notin K_1\cup K_2$ such that $\Phi(v)\in[b_2, a_1]$ is adjacent to some vertex in $\Phi^{-1}(a_1)\cap K_1$ and to some vertex in $\Phi^{-1}(b_2)\cap K_2$, and therefore, by homogeneity, is complete to $K_1\cup K_2$. Analogously, because of $\overline I_3$, each vertex $v\notin K_1\cup K_2$ such that $\Phi(v)\in [b_1, a_2]$ is complete to $K_1\cup K_2$. Therefore, the vertices of $V\setminus (K_1\cup K_2)$ that are not complete to $K_1\cup K_2$ are  in $(a_1,b_1)\cup (a_2,b_2)$, and therefore in $S_1\cup S_2$. Moreover, the vertices $v:\Phi(v)\in (a_1,b_1)$ are complete to each other (because they are in the interior of the interval $\overline I_1$) and similarly the vertices $v:\Phi(v)\in (a_2,b_2)$ are complete to each other: therefore, $S_1$ and $S_2$ are cliques. In order to show that $(K_1, K_2)$ is a fuzzy dominating pair, we are then left with proving that $S_1$ is complete to $S_3$ and $S_2$ is complete to $S_3$. Note that any vertex $v\notin K_1\cup K_2$ with $\Phi(v) \in [b_2,a_1]$ is complete to $\Phi^{-1}(a_1)$, because of $\overline I_4$, and therefore, by homogeneity, it is adjacent to every vertex in $\Phi^{-1}(b_1)\cap K_1$. Hence, there must exist an interval containing $[\Phi(v),b_1]$ (and not $[b_1,\Phi(v)]$, as we would contradict Claim \ref{claim:2}), and thus $v$ is complete to $S_1$ by definition. Similarly $v$ is complete to $S_2$. Using similar arguments, we can show that any vertex $v\notin K_1\cup K_2$ with $\Phi(v) \in [b_1,a_2]$ is complete to $S_1\cup S_2$. 
\bqed

\begin{claim}\label{claim:3}
For all $\overline I_1\in \setI$ such that $I'_1\cap \overline I_1\neq \emptyset$, we have $\Phi^{-1}(\overline I_1\cap I'_2) \subseteq K_2$, and, similarly, for all $\overline I_2\in \setI$ such that $I'_2\cap \overline I_2\neq \emptyset$, we have $\Phi^{-1}(\overline I_2\cap I'_1) \subseteq K_1$.
\end{claim}
Let us prove the first case. Suppose by contradiction that there exists an interval $\overline I_1\in \mathcal{I} : \overline I_1\cap I'_1\neq \emptyset$ covering $z \in \Phi^{-1}(I'_2) \setminus K_2$. Without loss of generality, let us assume that $b_1\in \overline I_1$ and that $[b_1, a_2]\subseteq [b_1, \Phi(z)] \subseteq \overline I_1$. In particular, $[b_1, a_2] \neq \overline I_1$, as otherwise we would contradict the assumption that $(\Phi, \mathcal{I})$ is not tight with respect to $(K_1, K_2)$.

First suppose that $z$ is adjacent to some vertex in $K_1 \cap \Phi^{-1}(b_1)$. Then, by homogeneity, there must exist an interval $L\in \mathcal{I}$ such that either $[\Phi(z),a_1]\subseteq L$ or $[a_1, \Phi(z)]\subseteq L$, but this latter case is ruled out by Claim \ref{claim:2}. Therefore, $[b_2, a_1]\subseteq [\Phi(z), a_1] \subseteq L$ and $L\neq [b_2,a_1]$, again by our assumptions. Summarizing, the following intervals belong to $\setI$: $I_1\supseteq [a_1,b_1]$; $I_2\supseteq [a_2,b_2]$; $\overline I_1\supsetneq [b_1,a_2]$; $L\supsetneq [b_2,a_1]$. But Claim \ref{claim:5} shows that this is a contradiction.

For the same reason, it follows that each vertex in $\Phi^{-1}(\Phi(z))\setminus K_2$ is anti-complete to $K_1 \cap \Phi^{-1}(b_1)$; but then $\overline I_1=[b_1, \Phi(z)]$ is a fuzzy interval. Now, because of our assumptions, there is no vertex $v\in K_2$ such that $\Phi(v) = \Phi(z)$. Therefore, each vertex $v\in K_1$ such that $\Phi(v) = b_1$ is anti-complete to $\Phi^{-1}(\Phi(z))$, a contradiction with the fact that each vertex at the extremity of a fuzzy interval is adjacent to some vertex at the other extremity.
\bqed

\smallskip
We are almost ready to build for $G$ our alternative representation $(\Phi',\setI')$ that is tight with respect to $(K_1, K_2)$. Note that, since $(K_1, K_2)$ is an almost-proper pair of cliques, there exists an edge $uv\in E$ with $u\in K_1,v\in K_2$. Now, since $u$ and $v$ are adjacent, there is an interval $J^*\in \setI$ covering $u$ and $v$. We can assume without loss of generality that $[b_1,a_2]\subseteq[\Phi(u),\Phi(v)]\subseteq J^*$. Let $\mathcal{J} \subseteq \mathcal{I}$ be the family of all intervals containing $[b_1,a_2]$: observe that each interval in $\mathcal{J}$ intersects $I'_1$, $I'_2$ but is neither containing $I'_1$ nor $I'_2$ (by Claim \ref{claim:2}) and does not cover any vertex $y \in \Phi^{-1}(I'_1) \cup \Phi^{-1}(I'_2)$  which is not in $K_1\cup K_2$  (by Claim \ref{claim:3}). We therefore define $l$ to be the closest extremity to $a_1$ in $I'_1$ of all the intervals in $\mathcal{J}$ and $r$ to be the closest extremity to $b_2$ in $I'_2$ of all the intervals in $\mathcal{J}$: note that $l\in (a_1, b_1]$ and $r\in [a_2, b_2)$ by Claim \ref{claim:2}. By definition, each interval $J\in \mathcal{J}$ is such that $J\subseteq [l,r]$. It follows from Claim \ref{claim:3} that $\Phi^{-1}([l, b_1]) \subseteq K_1$, $\Phi^{-1}([a_2, r]) \subseteq K_2$; moreover, $\Phi^{-1}((b_1, a_2)) \cap (K_1 \cup K_2) = \emptyset$.

We then define $\setI':=\setI \setminus \mathcal{J} \cup [l,r]$; $\Phi'(x)=\Phi(x)$ for all $x\in V\setminus(K_1\cup K_2)$, $\Phi'(x)=l$ for all $x\in K_1$ and $\Phi'(x)=r$ for all $x\in K_2$. We show in the following that the pair $(\Phi',\setI')$ defines the same adjacencies as the pair $(\Phi,\setI)$ and therefore that $(\Phi',\setI')$ is a representation of $G$ (note that no point of ${\mathcal C}$ is the end of more than one interval of $\setI'$ and no interval of $\setI'$ include another: this follows by construction and because $(\Phi,\setI)$ holds this property). Moreover $(\Phi',\setI')$ is tight with respect to $(K_1, K_2)$ by construction and, as it is easy to check, it can be built in time $O(n)$.

\smallskip

\begin{claim}\label{claim:4}
For any vertex $v\in V\setminus (K_1\cup K_2)$ such that there exists an interval of $\setI$ containing either $\Phi(v)$ and $a_1$, or  $\Phi(v)$ and $b_1$, there exists $I\in \setI$ such that $I$ contains $[a_1,b_1]$ and $\Phi(v)$.
\end{claim}
Indeed assume first that $\Phi(v)$ and $a_1$ are contained in an interval $J$ of $\setI$ and suppose the result does not hold. Observe that $\Phi(v)\not\in [a_1,b_1]$ otherwise the statement would hold with $I=I_1$. Necessarily $[\Phi(v),a_1]\subseteq J$ (else the statement holds again trivially). But then $\Phi(v) \not\in (b_1,a_2]$ because this would contradict Claim \ref{claim:2}. But we have also $\Phi(v) \not\in [a_2,r]$ as $\Phi^{-1}([a_2, r]) \subseteq K_2$. Thus $\Phi$ is in $(r,  a_1)$. If there exists $w\in \Phi^{-1}(\Phi(v))\setminus K_2$ adjacent to $\Phi^{-1}(a_1) \cap K_1$, then by homogeneity, $w$ is adjacent to $\Phi^{-1}(b_1)\cap K_1$ and since there does not exist an interval covering $[b_1,\Phi(v)]$ (it would contradict the definition of $r$ or Claim \ref{claim:2}), there is an interval covering $[\Phi(v),b_1]$, a contradiction. But then $\Phi^{-1}(a_1)\cap K_1$ is anti-complete to $\Phi^{-1}(\Phi(v))\setminus K_2$ and, in particular, neither $a_1$ nor $\Phi(v)$ is in the interior of $J$ and thus necessarily $J=[\Phi(v),a_1]$.  But $\Phi^{-1}(\Phi(v))\setminus K_2=\Phi^{-1}(\Phi(v))$ because else we contradict our assumption that $(\Phi,\setI)$ is not tight with respect to $(K_1,K_2)$. But now this contradicts our assumption that for every fuzzy interval in ${\mathcal I}$, every vertex mapped at one of the extremities has an adjacent and a non-adjacent vertex mapped at the other extremity. 

Suppose now that $\Phi(v)$ and $b_1$ belong to some interval $J$ of $\setI$ and suppose the result does not hold. Observe again that $\Phi(v)\not\in [a_1,b_1]$ otherwise the statement would hold with $I=I_1$ and thus again necessarily $[b_1,\Phi(v)]\subseteq J$ (else the statement holds again trivially). But $\Phi(v) \not\in [b_2,a_1)$ because of Claim \ref{claim:2}  and $\Phi(v) \not\in (r,b_2)$ by definition of $r$. We already observed that there is no vertex of $V\setminus (K_1\cup K_2)$ in $[a_2,r]$ thus $\Phi(v)\in (b_1,a_2)$. Now, because of interval $J^*$, $v$ is adjacent to $K_1$ and thus complete by homogeneity. Therefore there is an interval containing $a_1,\Phi(v)$ and this has to cover $[a_1,b_1]$ because otherwise this would contradict Claim \ref{claim:2}. But this is a contradiction.  
\bqed

We now show that the pair $(\Phi',\setI')$ defines the same adjacencies as the pair $(\Phi,\setI)$. We split the analysis of adjacencies into 3 cases: 

(i) Let $v$ be a vertex of $V\setminus (K_1\cup K_2)$, and first suppose it is complete to $K_1$. As a a consequence of Claim \ref{claim:4}, there exists an interval $\tilde I\in \setI$ containing $\Phi(v)$, $a_1$ and $b_1$, and, since this interval is not in $\mathcal{J}$, it also belongs to $\setI'$. Note that $l\in \tilde I$, and therefore $\Phi'(v)$ and $\Phi'(K_1)$ belong to $\tilde I\in \setI'$, meaning that we can preserve $v$ complete to $K_1$. Now suppose that $v$ is anti-complete to $K_1$. Note that $\Phi(v)\in (r,a_1)$, therefore, if $v$ is no more anti-complete to $K_1$, it is because $l$ and $\Phi(v)$ belong to some interval $\tilde I\in \setI$. It follows that either $[l, \Phi(v)]\subseteq \tilde I$, but this contradicts the definition of $r$, or $[\Phi(v), l]\subseteq \tilde I$, but then $v$ is not anti-complete to $K_1$, as $a_1$ is in the interior of $\tilde I$. The same holds for adjacencies between vertices of $V\setminus (K_1\cup K_2)$  and $K_2$. Thus adjacencies between vertices of $V\setminus (K_1\cup K_2)$  and $K_1\cup K_2$ are preserved.

(ii) Adjacencies between two vertices $u,v$ in $V\setminus (K_1\cup K_2)$ is unchanged because $\Phi'(u)=\Phi(u)$,  $\Phi'(v)=\Phi(v)$ and  $\Phi(u) , \Phi(v) \not\in [l,b_1] \cup [a_2,r]$. Indeed, if $uv \in E$, either there exists $J\not\in \mathcal J$ covering $u$ and $v$ or $\Phi(u),\Phi(v)\in (b_1,a_2)$, and in both cases adjacencies are preserved. Similarly, if not adjacent, at least one of $u$ or $v$ is in $(r,l)$ and thus no additional interval is added in $\setI'$ that could add the adjacency between $u$ and $v$.

(iii) Adjacencies between vertices in $K_1\cup K_2$ can be made arbitrary thanks to fuzziness of interval $[l,r]$.
\end{proof} 

\begin{theorem}
\label{main}
Let $G=(V, \Phi, \mathcal{I})$ be a connected fuzzy circular interval graph.
Let $(K_1, K_2)$ be an almost-proper and homogeneous pair of cliques. Then we can build in $O(n^2)$-time a representation $(\Phi'', \mathcal{I}')$, such that: $\Phi''(v) = a$, for each $v\in K_1$; $\Phi''(v) = b$, for each $v\in K_2$;  $[a,b]$ or $[b,a] \in \mathcal{I}'$, for some $a\neq b \in \mathcal C$. 
\end{theorem}

\begin{proof}
We know from Lemma \ref{fuzzy} that in time $O(n^2)$ we may build for $G$ a representation $(\Phi', \mathcal{I}')$ that is tight with respect to $(K_1, K_2)$, i.e. is such that, for some $u\in K_1$ and $v\in K_2$,  either $[\Phi(u), \Phi(v)]$ or $[\Phi(v), \Phi(u)]$ belongs to $\mathcal{I}$. We now show that we can build in $O(n)$-time from $\Phi'$ another mapping $\Phi''$ such that $(\Phi'',\setI')$ gives another representation for $G$ and satisfies the properties in the statement. Namely, define $\Phi''$ as follows: for every vertex $x \not\in K_1 \cup K_2$, let $\Phi''(x) = \Phi'(x)$;  for every vertex $x \in K_1$, let $\Phi''(x) = \Phi'(u)$;  for every vertex $x \in K_2$, let $\Phi''(x) = \Phi'(v)$.

In order to prove that $(\Phi'', \mathcal{I}')$ is a representation for $G$, it is enough to show that $(\Phi'', \mathcal{I}')$ and $(\Phi', \mathcal{I}')$ define the same adjacencies. In particular, it suffices to show that the neighborhood of every vertex $x$ such that $\Phi''(x) \neq \Phi'(x)$ remains the same: observe that such a vertex must belong to $K_1 \cup K_2$. Without loss of generality~choose $x \in K_1$. Now consider $y\in V\setminus x$. If $\Phi''(y) = \Phi'(v)$, then the adjacency between $x$ and $y$ is fuzzy in the new representation and of course we can preserve it. If $\Phi''(y) = \Phi'(u)$, then according to the new representation, $y$ is adjacent to $x$. We now show this to be correct. First, $y\not\in K_2$, since $\Phi''(y) \neq \Phi'(v)$. If $y\in K_1$, then adjacency between $y$ and $x$ follows from $K_1$ being a clique. If $y\not\in K_1\cup K_2$, then $\Phi'(y) = \Phi''(y) = \Phi'(u)$ and so $uy\in E$: adjacency between $x$ and $y$ follows then from homogeneity. Analogously, if $\Phi''(y) \not\in  \{\Phi'(v), \Phi'(u)\}$, then in particular $y\not\in K_1\cup K_2$ and thus the adjacency between $x$ and $y$ is the same as the adjacency between $u$ and $y$ (by homogeneity), which is preserved. Finally, $\Phi''$ can be built in time $O(n)$ from $\Phi'$.
\end{proof}

\section{A characterization for fuzzy circular interval graphs}
\label{sec:characterization}

In this section, we give a characterization for {\sc fcig}s. 
We start by giving the definition of an operation of {\em reduction} of a graph $G$ with respect to a homogenous pair of cliques (see Figure \ref{fig:reduction2}):

\begin{definition}
\label{reduction}
Let $(K_1, K_2)$ be a homogenous pair of cliques of $G$. The {\em reduction} of $G$ with respect to $(K_1, K_2)$ returns the graph $G|_{(K_1, K_2)}$ such that:
\begin{description}
\item $V(G|_{(K_1, K_2)}) = V(G)\cup \{x_1, y_1, x_2, y_2\} \setminus (K_1\cup K_2)$;
\item $E(G|_{(K_1, K_2)}) = \{uv: u,v\not\in K_1\cup K_2, uv\in E(G)\} \cup \{ux_1, ux_2: u\not\in K_1\cup K_2, u\in \Gamma(K_1)\} \cup \{uy_1, uy_2: u\not\in K_1\cup K_2, u\in \Gamma(K_2)\}\cup \{x_1x_2, y_1y_2, x_1y_1 \}$.
\end{description}
\end{definition}

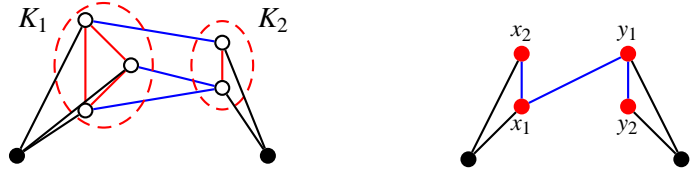
\begin{figure}[tbh]
\begin{center}
\vspace{.5cm}
\begin{tabular}{ccc}
\psset{unit=0.6cm}
\begin{pspicture}(4,0.5)(5,2)
\cnode(1,1){3pt}{1}
\cnode(2,2){3pt}{2}
\cnode(1,3){3pt}{3}
\uput[l](0.5,3){$K_1$}

\ncline[linecolor=red]{1}{2}
\ncline[linecolor=red]{2}{3}
\ncline[linecolor=red]{1}{3}

\cnode(4,1.5){3pt}{4}
\cnode(4,2.5){3pt}{5}
\uput[r](4.5,3){$K_2$}

\ncline[linecolor=red]{4}{5}

\ncline[linecolor=blue]{1}{4}
\ncline[linecolor=blue]{2}{4}
\ncline[linecolor=blue]{3}{5}

\psellipse[linecolor=red, linestyle=dashed](1.4,2)(1.1,1.4)
\psellipse[linecolor=red, linestyle=dashed](4,2)(.7,1)

\cnode*(-.5,0){3pt}{6}
\cnode*(5,0){3pt}{7}

\ncline[linecolor=black]{6}{1}
\ncline[linecolor=black]{6}{2}
\ncline[linecolor=black]{6}{3}
\ncline[linecolor=black]{7}{4}
\ncline[linecolor=black]{7}{5}

\end{pspicture} &&

\hspace*{1cm}

\psset{unit=0.7cm}
\begin{pspicture}(4,0.5)(5,2)
\footnotesize
\cnode*[linecolor=red](6,1){3pt}{1}
\uput[d](6,1){$x_1$}

\cnode*[linecolor=red](6,2){3pt}{1'}
\uput[u](6,2){$x_2$}

\cnode*[linecolor=red](8,1){3pt}{2}
\uput[d](8,1){$y_2$}

\cnode*[linecolor=red](8,2){3pt}{2'}
\uput[u](8,2){$y_1$}

\cnode*(5,0){3pt}{3}
\cnode*(9,0){3pt}{4}

\ncline[linecolor=blue]{1}{1'}
\ncline[linecolor=blue]{2}{2'}
\ncline[linecolor=blue]{1}{2'}

\ncline[linecolor=black]{1}{3}
\ncline[linecolor=black]{1'}{3}

\ncline[linecolor=black]{2}{4}
\ncline[linecolor=black]{2'}{4}

\end{pspicture} \\
\end{tabular}
\end{center}
\caption{A proper and homogeneous pair of cliques $(K_1,K_2)$ (on the left) and the reduction of the graph with respect to the pair $(K_1,K_2)$ (on the right).} 
\label{fig:reduction2}
\end{figure}

In Figure \ref{fig:reduction2}, the reduction of a (proper and) homogeneous pair of cliques is depicted.
We skip the proof of the following simple lemma:

\begin{lemma}
\label{reduce_w}
Let $G$ be a connected graph and $(K_1, K_2)$ a homogeneous pair of cliques. Then $G|_{(K_1, K_2)}$ is connected.
\end{lemma}

We are now ready to prove our main result, which  shows that the reduction of proper and homogeneous pairs of cliques preserves the property of a graph to be fuzzy circular interval.

\begin{theorem}
\label{characterization}
Let $G$ be a connected graph and let $(K_1, K_2)$ be a proper and homogeneous pair of cliques. Then $G|_{(K_1, K_2)}$ is connected. Moreover, $G$ is a fuzzy circular interval graph if and only if $G|_{(K_1, K_2)}$ is a fuzzy circular interval graph and, from a representation for $G$, one may build in $O(n^2)$-time a representation for $G|_{(K_1, K_2)}$, and vice versa.
\end{theorem}

\begin{proof}
From Lemma \ref{reduce_w}, $G|_{(K_1, K_2)}$ is connected. We now show that $G$ is a {\sc fcig} if and only if $G|_{(K_1, K_2)}$ is a {\sc fcig}. 

\smallskip
{\em Necessity}. 
From Theorem \ref{main}, we know that, from any representation for $G$, we may build in $O(n^2)$-time another one $(\Phi, \mathcal{I})$ such that $\Phi(K_1) = a$, $\Phi(K_2) = b$ for some $a\neq b \in \mathcal C$ and, without loss of generality, $[a,b] \in \mathcal{I}$. Consider the following mapping $\Phi'$ for the vertices of $G|_{(K_1, K_2)}$ (cfr. Definition \ref{reduction}):

\smallskip $\bullet$  for $v \in V(G|_{(K_1, K_2)}) \setminus \{x_1, y_1, x_2, y_2\}$, $\Phi'(v) = \Phi(v)$;

\smallskip $\bullet$  for $v\in \{x_1, x_2\}$, $\Phi'(v) = a$;

\smallskip $\bullet$  for $v\in \{y_1, y_2\}$, $\Phi'(v) = b$.

\smallskip 
We claim that $(\Phi', \mathcal{I})$ is a representation for $G|_{(K_1, K_2)}$, i.e.~that $(\Phi', \mathcal{I})$ is consistent with $E(G|_{(K_1, K_2)})$. First, consider $u$ and $v\in \{x_1, y_1, x_2, y_2\}$. In this case, consistency holds since $[a,b]$ is an interval of $\mathcal{I}$. Now consider $u$ and $v\notin \{x_1, y_1, x_2, y_2\}$. In this case, $uv\in E(G|_{(K_1, K_2)})$ if and only if $uv\in E(G)$: consistency follows since $\Phi'(v) = \Phi(v)$, $\Phi'(u) = \Phi(u)$ and we keep $\mathcal{I}$. Finally consider $u$ and $v$ such that $u\in \{x_1, y_1, x_2, y_2\}$, e.g.~$u\in\{x_1, x_2\}$, and $v\not\in \{x_1, y_1, x_2, y_2\}$. In this case, $uv\in E(G|_{(K_1, K_2)})$ if and only if $v\in \Gamma(K_1)$: consistency follows since $\Phi'(v) = \Phi(v)$, $\Phi'(u) = \Phi(K_1)$ and we keep $\mathcal{I}$.

\smallskip
{\em Sufficiency}. Note that $G|_{(K_1, K_2)}$ satisfies the hypothesis of Theorem \ref{main} with $(\{x_1, x_2\},\{y_1, y_2\})$ being an almost-proper homogeneous pair of cliques. Therefore, from any representation for $G|_{(K_1, K_2)}$, we may build in $O(n^2)$-time another representation $(\Phi, \mathcal{I})$ such that such that $\Phi(v) = a$ for $v\in \{x_1, x_2\} $, $\Phi(v) = b$ for $v\in \{y_1, y_2\}$ and without loss of generality~$[a,b] \in  \mathcal{I}$. In order to show that $G$ is fuzzy circular interval too, we consider the pair $(\Phi', \mathcal{I})$, where $\Phi'$ is such that: 

\smallskip $\bullet$   for $v \in V(G) \setminus (K_1 \cup K_2)$, $\Phi'(v) = \Phi(v)$;

\smallskip $\bullet$  for $v\in K_1$, $\Phi'(v) = a$;

\smallskip $\bullet$  for $v\in K_2$, $\Phi'(v) = b$.

\smallskip

It is again easy to show that $(\Phi', \mathcal{I})$ is a representation for $G$, we omit the details.
\end{proof}

\section{Recognizing fuzzy circular interval graphs}
\label{sec:recognition} 

In order to provide our recognition algorithm for {\sc fcig}s, we need a result from the literature. It is a natural algorithm for finding a proper and homogenous pair of cliques, that appears in King and Reed \cite{kr} and in Pietropaoli \cite{UgoPhd}. A vertex $v$ of a graph $G(V,E)$ is \emph{universal} to $u \in V$ if $v$ is adjacent to $u$ and to every vertex in $N(u) \setminus \{v\}$.

\begin{algorithm}
\caption{Finding proper and homogeneous pairs of cliques}
\label{PH}
\begin{algorithmic} [1]
\REQUIRE A graph $G$.
\ENSURE A proper and homogeneous pair of cliques $(K', K)$, if any.

\smallskip
\FOR {each pair of adjacent vertices $\{u, v\}$, that are not universal to each other}
\STATE  $K' := \{u, v\}$; $K := P(\{u,v\})$.
\WHILE {$K$ is a clique and $P(K) \neq K'$}
\STATE{$K'\leftarrow K$, $K\leftarrow P(K)$}
\ENDWHILE
\STATE {\bf if} {$K$ is not a clique} {\bf then} {there is no proper and homogeneous pair of cliques $(K_1, K_2)$ such that either $\{u,v\} \subseteq K_1$ or $\{u,v\} \subseteq K_2$.}
\STATE {\bf else} {$P(K) =K'$} and {$\{K,K'\}$ is a proper and homogeneous pair of cliques: {\bf stop}}
\ENDFOR
\end{algorithmic}
\end{algorithm}

We are now ready to state our recognition algorithm (Algorithm \ref{reco}) for {\em connected} graphs (we shall take care of non-connected graphs later).

\begin{algorithm}[h]
\caption{The recognition algorithm}
\label{reco}
\begin{algorithmic} [1]
\REQUIRE A connected graph $G$.
\ENSURE Say whether $G$ is fuzzy circular interval and, in case, find a representation.

\smallskip
\STATE $i = 0$; $G^0 = G$.
\WHILE {$G^i$ has a proper and homogeneous pair of cliques $(X_i, Y_i)$}
\STATE { $G^{i+1}: = G^{i}|_{(X_i, Y_i)}$; $i = i+1$.} \ENDWHILE
\STATE  $q:= i$.
\IF{$G^{q}$ is not a circular interval graph}
\STATE {$G$ is not a fuzzy circular interval graph: {\bf stop}.} 
\ELSE  
\STATE Compute a (fuzzy) interval representation for $G^{q}$.
\FOR {$h=q$ down to 1} 
\STATE {extend the representation for $G^{h}$ into a representation for $G^{h-1}$ using Th. \ref{characterization}.} \ENDFOR
\ENDIF
\end{algorithmic}
\end{algorithm}

\begin{theorem}
\label{complexity}
Algorithm \ref{reco} is correct and terminates in at most $m$ iterations.
\end{theorem}

\begin{proof}
The algorithm defines a sequence of graphs $G^0$, \ldots, $G^q$, for some $q\leq m$: in fact, each proper and homogeneous pair of cliques contains $C_4$ (see Lemma \ref{building_block}), that has 4 edges, as an induced subgraph, while the gadget we use in our reduction operation has 3 edges. It also claims that $G = G^0$ is a {\sc fcig} if and only if  $G^q$ is a circular interval graph. That is correct. In fact, on the one hand, Theorem \ref{characterization} ensures that each graph in the sequence is a connected graph, that is a {\sc fcig} if and only if $G$ is so. On the other hand, since $G^q$ is a graph without proper and homogeneous pairs of cliques, from Lemma \ref{fuzzy2} it is a {\sc fcig} if and only if it is a circular interval graph.

Moreover, if $G$ is a {\sc fcig}, then the algorithm returns a representation for it. In fact, in this case, $G^q$ is a circular interval graph and the algorithm computes a (fuzzy) interval representation for it; this representation can be then extended onto representations for $G^{q-1}, \ldots, G^0$ following Theorem \ref{characterization}.
\end{proof}

{\bf Complexity issues.} We now analyze the complexity of Algorithm \ref{reco}. As it is shown in \cite{kr}, it is possible to implement Algorithm \ref{PH} as to run in $O(n^2m)$-time; therefore, for each $i = 0, \ldots, q-1$, a proper and homogeneous pair of cliques of $G^i$ can be found in O$(n^2m)$-time.  As we discussed in Section \ref{fcig} we can recognize if a graph with $n_1$ vertices and $m_1$ edges is circular interval, and in case build an interval representation, in $O(n_1+m_1)$-time, with a trivial modification of the algorithm in \cite{DHH}; therefore we can recognize in $O(n+m)$-time whether $G^{q}$ is a circular interval graph (in fact, as we already discussed $m_1 < m$, moreover, $n_1\leq n$). Each graph $G^i$ can be built from the graph $G^{i-1}$ in linear time and each representation for $G^{i-1}$ can be extended into a representation for $G^i$ in time $O(n^2)$ (because of Theorem \ref{characterization} and since $|V(G^i)|\leq n$). Since the number of iterations is bounded by $m$, it easily follows that Algorithm \ref{reco} can be indeed implemented as to run in $O(n^2m^2)$-time.

Even better, building upon some arguments from \cite{FOS}, it is possible to show that this complexity can be lowered to $O(n^2m)$. The crucial fact is the following. Say that two adjacent vertices $u$ and $v$ of $G$ form a {\em candidate} pair if they are not universal to each other; thus, the candidate pairs of $G$ are at most the number of its edges. Note that Algorithm \ref{PH} receives a candidate pair in input. In \cite{FOS} it is shown that, when passing from $G$ to $G|_{(K_1, K_2)}$, at least one candidate pair is destroyed, and no new candidate pair is created. It follows that, throughout all the iterations of Algorithm \ref{reco}, each candidate pair is scanned at  most once. With some care, this observation leads to an implementation of Algorithm \ref{reco} running in $O(n^2m)$-time. More details can be found in \cite{FOS}.

We summarize the previous discussion into the following: 

\begin{lemma}
\label{complexitybis}
It is possible to implement Algorithm \ref{reco} as to run in $O(n^2m)$-time.
\end{lemma}

\smallskip
We close the paper by discussing what to do when $G$ is not connected. In this case, we have the following simple lemma, whose proof we omit.

\begin{lemma}
\label{connection}
Let $G=(V, \Phi, \mathcal{I})$ be a non-connected graph. $G$ is a fuzzy circular interval graph if and only if each connected component is a fuzzy linear interval graph. 
\end{lemma}

Therefore, the problem of recognizing non-connected {\sc fcig}s reduces to the problem of recognizing (connected) fuzzy linear interval graphs. Also the latter problem can be solved by our reduction techniques; we have in fact the following:

\begin{theorem}
\label{characterization_bis}
Let $G$ be a connected graph and let $(K_1, K_2)$ be a proper and homogeneous pair of cliques. $G$ is a fuzzy linear interval graph if and only if $G|_{(K_1, K_2)}$ is a fuzzy linear interval graph  and, from a representation for $G$, one may build in $O(n^2)$-time a representation for $G|_{(K_1, K_2)}$, and vice versa.
\end{theorem}

The proof of  Theorem \ref{characterization_bis} goes along the same lines as the proof for Theorem \ref{characterization} so we skip it. Finally, Theorem \ref{characterization_bis} and Lemma \ref{fuzzy2} reduce the recognition of non-connected {\sc fcig}s to that of linear interval graphs. The latter problem can be easily solved \cite{Cetal, HMM, PD}.

\end{document}